\documentclass[11pt,letterpaper]{article}
\usepackage{amsmath,amsfonts,amssymb,amsthm,graphicx,epic,eepic,epsfig,latexsym,color,revsymb}

\def\squareforqed{\hbox{\rlap{$\sqcap$}$\sqcup$}}
\def\qed{\ifmmode\squareforqed\else{\unskip\nobreak\hfil
\penalty50\hskip1em\null\nobreak\hfil\squareforqed
\parfillskip=0pt\finalhyphendemerits=0\endgraf}\fi}
\def\endenv{\ifmmode\;\else{\unskip\nobreak\hfil
\penalty50\hskip1em\null\nobreak\hfil\;
\parfillskip=0pt\finalhyphendemerits=0\endgraf}\fi}

\newtheorem{theorem}{Theorem}

\newtheorem{definition}[theorem]{Definition}

\newtheorem{lemma}[theorem]{Lemma}

\newtheorem{proposition}[theorem]{Proposition}

\newcommand{\nc}{\newcommand}
\nc{\rnc}{\renewcommand}
\nc{\bra}[1]{\langle#1|}
\nc{\ket}[1]{|#1\rangle}
\nc{\ketbra}[2]{|#1\rangle\!\langle #2|}
\nc{\braket}[2]{\langle #1 | #2  \rangle}
\nc{\proj}[1]{| #1\rangle\!\langle #1 |}
\nc{\avg}[1]{\langle#1\rangle}
\nc{\sfrac}[2]{\mbox{$\frac{#1}{#2}$}}
\nc{\ox}{\otimes}
\nc{\dg}{\dagger}
\nc{\lbar}[1]{\overline{#1}}
\nc{\rar}{\rightarrow}
\nc{\dn}{\downarrow}
\nc{\lrar}{\longrightarrow}
\nc{\tr}{\operatorname{Tr}}
\nc{\supp}{{\operatorname{supp}}}
\nc{\Proj}{{\operatorname{proj}}}
\nc{\var}{\operatorname{var}}
\nc{\expect}{\operatorname{E}}
\nc{\Rank}{\operatorname{Rank}}
\nc{\polylog}{\operatorname{polylog}}
\nc{\id}{{\operatorname{id}}}
\nc{\err}{\operatorname{Err}}
\nc{\1}{{\openone}}
\nc{\di}{\mathrm{d}}

\nc{\cA}{\mathcal{A}}  \nc{\cB}{\mathcal{B}}  \nc{\cC}{\mathcal{C}}
\nc{\cD}{\mathcal{D}}  \nc{\cE}{\mathcal{E}}  \nc{\cF}{\mathcal{F}}
\nc{\cG}{\mathcal{G}}  \nc{\cH}{\mathcal{H}}  \nc{\cI}{\mathcal{I}}
\nc{\cJ}{\mathcal{J}}  \nc{\cK}{\mathcal{K}}  \nc{\cL}{\mathcal{L}}
\nc{\cM}{\mathcal{M}}  \nc{\cN}{\mathcal{N}}  \nc{\cO}{\mathcal{O}}
\nc{\cP}{\mathcal{P}}  \nc{\cQ}{\mathcal{Q}}  \nc{\cS}{\mathcal{S}}
\nc{\cT}{\mathcal{T}}  \nc{\cX}{\mathcal{X}}  \nc{\cZ}{\mathcal{Z}}

\nc{\RR}{{{\mathbb R}}} \nc{\CC}{{{\mathbb C}}} \nc{\FF}{{{\mathbb F}}}
\nc{\NN}{{{\mathbb N}}} \nc{\ZZ}{{{\mathbb Z}}} \nc{\PP}{{{\mathbb P}}}
\nc{\QQ}{{{\mathbb Q}}} \nc{\UU}{{{\mathbb U}}} \nc{\EE}{{{\mathbb E}}}

\begin{document}

\title{Discriminating quantum states: the multiple Chernoff distance}
\author{Ke Li \\ \\IBM T.J. Watson Research Center \\ and Massachusetts Institute of Technology}
\date{}
\maketitle

\begin{abstract}
  We consider the problem of testing multiple quantum hypotheses $\{\rho_1^{\ox n},\ldots,
  \rho_r^{\ox n}\}$, where an arbitrary prior distribution is given and each of the $r$
  hypotheses is $n$ copies of a quantum state. It is known that the minimal average error
  probability $P_e$ decays exponentially to zero, that is, $P_e=\exp\{-\xi n+o(n)\}$.
  However, this error exponent $\xi$ is generally unknown, except for the case that $r=2$.

  In this paper, we solve the long-standing open problem of identifying the above error
  exponent, by proving Nussbaum and Szko{\l}a's conjecture that $\xi=\min_{i\neq j}C(\rho_i,
  \rho_j)$. The right-hand side of this equality is called the multiple quantum Chernoff
  distance, and $C(\rho_i,\rho_j):=\max_{0\leq s\leq 1}\{-\log\tr\rho_i^s\rho_j^{1-s}\}$
  has been previously identified as the optimal error exponent for testing two hypotheses,
  $\rho_i^{\ox n}$ versus $\rho_j^{\ox n}$.

  The main ingredient of our proof is a new upper bound for the average error probability,
  for testing an ensemble of finite-dimensional, but otherwise general, quantum states.
  This upper bound, up to a states-dependent factor, matches the multiple-state generalization
  of Nussbaum and Szko{\l}a's lower bound. Specialized to the case $r=2$, we give an
  alternative proof to the achievability of the binary-hypothesis Chernoff distance, which
  was originally proved by Audenaert et al.
\end{abstract}

\begin{table}[b]
\rule{3cm}{0.01cm} \newline
\par
{\small Email: carl.ke.lee@gmail.com}
\newline{\small  Supported by NSF grants CCF-1110941 and CCF-1111382.}
\end{table}

\section{Introduction}
   \label{section:introduction}
A basic problem in information theory and statistics, is to test a system that may be prepared
in one of $r$ random states. Treated in the framework of quantum mechanics, the testing
is performed via quantum measurement, and the physical states are described by density matrices
$\omega_1, \omega_2,\ldots, \omega_r$, namely, positive semidefinite Hermitian matrices of
trace 1. It is a notable fact that, when $\omega_i\text{'s}$ commute, the problem reduces to
classical statistical testing, among $r$ probability distributions that are given by the arrays
of eigenvalues of each of the density matrices. However, the generally noncommutative feature
makes quantum statistics much richer than its classical counterpart.

Our main focus in the current paper will be on the asymptotic setting. Let the tensor product
state $\rho^{\ox n}$ denotes $n$ independent copies of $\rho$, in analogy to the probability
distribution of i.i.d. random variables. We are interested in the asymptotic behavior of the
average error $P_e$, in discriminating a set of quantum states $\{\rho_1^{\ox n},\ldots,
\rho_r^{\ox n}\}$, when an arbitrary prior that is independent of $n$ is given. Parthasarathy
showed that $P_e$ decays exponentially, that is, $P_e=\exp\{-\xi n+o(n)\}$~\cite{Parth01}.
However, to date the optimal error exponent $\xi$, as a functional of the states $\rho_1,\dots,
\rho_r$, is generally unknown.

Significant achievements have been made for the case of testing two quantum hypotheses ($r=2$).
In two breakthrough papers, \cite{ACMMABV} and~\cite{NussbaumSzkola06}, it has been established
that the optimal error exponent in discriminating $\rho_1^{\ox n}$ and $\rho_2^{\ox n}$, equals
the quantum Chernoff distance
\[C(\rho_1,\rho_2):=\max_{0\leq s\leq 1}\{-\log\tr\rho_1^s\rho_2^{1-s}\}.
\]
Audenaert et al in~\cite{ACMMABV} solved the achievability part, in the meantime Nussbaum
and Szko{\l}a in~\cite{NussbaumSzkola06} proved the optimality part. This provides
the quantum generalization of the Chernoff information as the optimal error exponent in classical
hypotheses testing~\cite{Chernoff52}; see also~\cite{CT91}.

The solution for the general cases $r>2$ is still lacking and it does not follow from the
binary case directly. The optimal tests, as analogs of the classical maximum likelihood decision
rule, have been formulated in the 1970s. For discriminating two states it has an explicit
expression known as the Holevo?Helstrom test~\cite{Helstrom76, Holevo78}, and indeed, the proof
in~\cite{ACMMABV} relies on a nontrivial application of this Holevo?Helstrom test. In contrast,
for discriminating multiple quantum states the corresponding optimal measurement can only be
formulated in a very complicated, implicit way~\cite{Holevo73, YKL75}. Such a situation illustrates
the difficulty in dealing with the asymptotic error exponent, for the multiple case $r>2$.
Intuitively, competitions among pairs make the problem complicated.

Nussbaum and Szko{\l}a introduced the multiple quantum Chernoff distance
\[C(\rho_1,\ldots,\rho_r):=\min_{(i,j):i\neq j}C(\rho_i,\rho_j),
\]
and conjectured that it is the optimal asymptotic error exponent, in discriminating quantum
states $\rho_1^{\ox n},\ldots, \rho_r^{\ox n}$~\cite{NussbaumSzkola10,NussbaumSzkola11a,NussbaumSzkola11}.
This is in full analogy to the existing results in classical statistical hypothesis
testing~\cite{LeangJohnson97,Salikh73,Salikh98,Torgersen81}. Significant progress has been
made towards proving this conjecture. Besides the case of commuting states which reduces to the
classical situation, it has been proven to be true in several interesting special cases. These
include when the supporting spaces of the states $\rho_1,\ldots,\rho_r$ are disjoint~\cite{NussbaumSzkola11},
and when one pair of the states is substantially closer than the other pairs, in Chernoff
distance~\cite{Nussbaum13,Aud-Mosonyi14}. In general, Nussbaum and Szko{\l}a showed that the
optimal error exponent $\xi$ in testing multiple quantum hypotheses, satisfies
$C/3\leq\xi\leq C$~\cite{NussbaumSzkola11}, and Audenaert and Mosonyi recently strengthened
this bound, showing that $C/2\leq\xi\leq C$~\cite{Aud-Mosonyi14}.

In this paper, we shall prove the aforementioned conjecture, that is, we show that the long-sought
error exponent in asymptotic quantum (multiple) state discrimination, is given by the (multiple)
quantum Chernoff distance. Besides, as a main ingredient of the proof we derive a new upper bound
for the optimal average error probability, for discriminating a set of finite-dimensional, but
otherwise general, quantum states. This one-shot upper bound has the advantage that, up to a
states-dependent factor, it coincides with a multiple-state generalization of Nussbaum and Szko\l a's
lower bound~\cite{NussbaumSzkola06}.

Before concluding this section, we review the relevant literature.
Asymptotics of statistical hypothesis testing is an important topic in statistics and information
theory, and is especially useful in identifying basic information quantities. We refer the
interested readers for a partial list of classical results to~\cite{Blahut74, Chernoff52, CT91,
CsiszarLongo71, HanKobayashi89, Hoeffding65, LeangJohnson97}, and quantum results to~\cite{ACMMABV,
ANSV08,BDKSSS05,BrandaoPlenio09,HiaiPetz91,Li12,MosonyiOgawa13,NussbaumSzkola06, OgawaNagaoka00,
TomamichelHayashi12}. The optimal or approximately optimal average error in quantum state
discrimination, and the corresponding tests to achieve it, is a basic problem in quantum
information theory and has attracted extensive study; see, for example, \cite{Aud-Mosonyi14,
BarnumKnill02, HausladenWootters94, Helstrom76, Holevo73, Holevo78, KRS09, Qiu08, Tyson09, YKL75}.

The remainder of this paper is organized as follows. After introducing some basic notations, concepts
and the relevant aspects of the quantum formalism in Section~\ref{section:notations-preliminaries},
we present the main results in Section~\ref{section:results}. Section~\ref{section:proofs} is
dedicated to the proofs. At last, in Section~\ref{section:discussions}, we conclude the paper with
some discussion and open questions.

\section{Notation and preliminaries}
  \label{section:notations-preliminaries}
Let $\cB(\cH)$ denote the set of linear operators on a complex, finite-dimensional Hilbert space
$\cH$. Let $\cP(\cH)\subset\cB(\cH)$ be the set of positive semidefinite matrices, and $\cD(\cH)
:=\{\omega:\omega\in\cP(\cH), \text{and} \tr\omega=1\}$ is the set of density matrices. We say a
matrix $A\geq 0$ if $A\in\cP(\cH)$, and $A\geq B$ if $A-B\geq 0$. The dimension of the Hilbert
space $\cH$ is denoted as $|\cH|$. $\1$ denotes the identity matrix. We use the Dirac notation
$\ket{v}\in\cH$ to denote a unit vector, $\bra{v}$ its conjugate transpose, and $\braket{v}{w}$
the inner product. A Hermitian matrix $X$ can be written in the spectral decomposition form:
$X=\sum_i\lambda_iQ_i$, where $\lambda_i\text{'s}$ satisfying $\lambda_i\neq\lambda_j$ for $i\neq j$
are the eigenvalues, and $Q_i\text{'s}$ satisfying $Q_iQ_j=\delta_{ij}Q_i$ and $\sum_iQ_i=\1$ are
the orthogonal projectors onto the eigenspaces. $\supp(X)$ is the supporting space of $X$ and is
spanned by all the eigenvectors with non-zero eigenvalues, $\{X>0\}:=\sum_{i:\lambda_i>0}Q_i$
represents the projector onto the positive supporting space of $X$, and $\Omega(X):=|\{\lambda_i
\}_i|$ denotes the number of eigenspaces, or distinct eigenvalues. For a subspace $S\subset\cH$,
$\Proj(S)$ is the projector onto $S$. The sum of two subspaces $S_1,S_2\subset\cH$, is defined as
$S_1+S_2:=\{u+v|u\in S_1, v\in S_2\}$. When we say the overlap between two subspaces $S_1$ and
$S_2$, we mean the maximal overlap between two unit vectors from each of them: $\max\{|\braket{v}
{w}|:\ket{v}\in S_1,\ket{w}\in S_2\}$.

We briefly review some aspects of the quantum formalism, relevant in this paper. Every physical
system is associated with a complex Hilbert space, which is called the state space. The states
of a system are described by density matrices. Pure states are of particular interest, and are
represented by rank-one projectors, or simply the corresponding unit vectors. Throughout this
paper, we are concerned with quantum states of a finite system, associated with a finite-dimensional
Hilbert space. A density matrix $\omega$ can be decomposed as the sum of an ensemble of pure states,
that is, $\omega=\sum_ip_i\proj{\psi_i}$, with $\{p_i\}$ a probability distribution. An intuitive
understanding is that pure states represent ``deterministic events'', and a density matrix is the
quantum analogue of a probability distribution over these events. However, note that this decomposition
is not unique, and non-orthogonal pure states are not perfectly distinguishable.

The procedure of detecting the state of a quantum system is called quantum measurement, which,
in the most general form, is formulated as positive operator-valued measure (POVM), that is,
$\cM=\{M_i\}_i$, with the POVM elements satisfying $0\leq M_i \leq \1$ and $\sum_i M_i=\1$. When
performing the measurement on a system in the state $\omega$, we get outcome $i$ with
probability $\tr(\omega M_i)$. Projective measurements, or von Neumann measurements, are special
situations of POVMs, where all the POVM elements are orthogonal projectors: $M_iM_j=\delta_{ij}M_i$,
with $\delta_{ij}$ the Kronecker delta function.

Suppose a physical system, also called an information source, is in one of a finite set of hypothesized
states $\{\omega_1,\ldots, \omega_r\}$, with a given prior $\{p_1,\ldots, p_r\}$. For convenience,
we denote them as a normalized ensemble $\{A_1:=p_1\omega_1,\ldots, A_r:=p_r\omega_r\}$. To determine
the true state, we make a POVM measurement $\{M_1,\ldots, M_r\}$, and infer that it is in the state
$\omega_i$ if we get outcome $i$. The average (Bayesian) error probability is
\begin{equation}
  \label{eq:error-definition}
  P_e\left(\{A_1,\ldots,A_r\}; \{M_1,\ldots, M_r\}\right):=\sum_{i=1}^r\tr A_i(\1-M_i).
\end{equation}
Minimized over all possible measurements, this gives the optimal error probability
\begin{equation}
  \label{eq:optimal-error-definition}
  P_e^*\left(\{A_1,\ldots, A_r\}\right)\!:=
           \min\!\left\{\!\sum_{i=1}^r\tr A_i(\1-M_i)\!:\text{POVM }\! \{M_1,\ldots, M_r\}\!\! \right\}.
\end{equation}
We note here that the definitions~(\ref{eq:error-definition}) and (\ref{eq:optimal-error-definition})
apply, as well, to a non-normalized ensemble of quantum states $\{A_1,\ldots,A_r\}$ which only
satisfies the constraint $(\forall i)\, A_i\geq 0$. In this case, $P_e$ and $P_e^*$ may not have a
clear meaning but sometimes can be useful.

In the asymptotic setting where $\omega_i$ is replaced by the tensor product state $\rho_i^{\ox n}$,
we are interested in the behavior of the optimal error $P_e^*$, as $n\rightarrow \infty$. An important
quantity characterizing this asymptotic behavior, is the rate of exponential decay, or simply
\emph{error exponent}
\[\liminf_{n\rar\infty}\frac{-1}{n}\log P_e^*\left(\{p_1\rho_1^{\ox n},\ldots,p_r\rho_r^{\ox n}\}\right).
\]

\section{Results}
  \label{section:results}
Our main result is the following Theorem~\ref{thm:quantum-multiple-Chernoff}. Recall that, for
the case $r=2$ of testing two hypotheses, it has been proven nearly a decade ago in
2006; see~\cite{ACMMABV} and~\cite{NussbaumSzkola06}.

\begin{theorem}
  \label{thm:quantum-multiple-Chernoff}
  Let $\{\rho_1,\ldots,\rho_r\}$ be a finite set of quantum states on a finite-dimensional Hilbert
  space $\cH$. Then the asymptotic error exponent for testing $\{\rho_1^{\ox n},\ldots,\rho_r^{\ox n}\}$,
  for an arbitrary prior $\{p_1,\ldots,p_r\}$, is given by the multiple quantum Chernoff distance:
  \begin{equation}
    \label{eq:error-exponent}
    \lim_{n\rar\infty}\frac{-1}{n}\log P_e^*\left(\{p_1\rho_1^{\ox n},\ldots,p_r\rho_r^{\ox n}\}\right)
    \!=\!\min_{(i,j):i\neq j}\max_{0\leq s\leq 1}\left\{\!-\log\tr\rho_i^s\rho_j^{1-s}\!\right\}.
  \end{equation}
\end{theorem}

The optimality part, that is, the left-hand side of equation~(\ref{eq:error-exponent}) being upper
bounded by the right-hand side, follows easily from the optimality of the binary case
$r=2$~\cite{NussbaumSzkola06}; see the argument in~\cite{NussbaumSzkola11a}. Roughly speaking, this is
because, discriminating an arbitrary pair within a set of quantum states is easier than discriminating
all of them. On the other hand, the achievability part is the main difficulty in proving
Theorem~\ref{thm:quantum-multiple-Chernoff}. In~\cite{ACMMABV}, Audenaert et~al employed the
Holevo?Helstrom tests $\left(\{\rho_1^{\ox n}-\rho_2^{\ox n}>0\}, \1-\{\rho_1^{\ox n}-\rho_2^
{\ox n}>0\}\right)$ to achieve the binary Chernoff distance in testing $\rho_1^{\ox n}$ versus
$\rho_2^{\ox n}$. However, to date we do not have a way to generalize the method of Audenaert et~al
to deal with the $r>2$ cases, even though there is the multiple generalization of the Holevo?Helstrom
tests~\cite{Holevo73, YKL75}; see discussions in \cite{NussbaumSzkola11} and \cite{Aud-Mosonyi14}
on this issue. Here, using a conceptually different method, we derive a new upper bound for the
optimal error probability of equation~(\ref{eq:optimal-error-definition}). This one-shot error bound,
as stated in Theorem~\ref{thm:one-shot-achi}, works for testing any finite number of finite-dimensional
quantum states, and when applied in the asymptotics for i.i.d. states, accomplishes the achievability
part of Theorem~\ref{thm:quantum-multiple-Chernoff}.

Our method is inspired by the previous work of Nussbaum and Szko{\l}a~\cite{NussbaumSzkola11}.
It is shown in~\cite{NussbaumSzkola11} that if the supporting spaces of the hypothetic states
$\rho_1,\ldots,\rho_r$ are pairwise disjoint (this means that the supporting spaces of
$\rho_1^{\ox n},\ldots,\rho_r^{\ox n}$ are asymptotically highly orthogonal), then the Gram-Schmidt
orthonormalization can be employed to construct a good measurement, which achieves the error
exponent of Theorem~\ref{thm:quantum-multiple-Chernoff}. Here to prove
Theorem~\ref{thm:quantum-multiple-Chernoff} for general hypothetic states, we
find a way to remove a subspace from each eigenspace of the states $\rho_1^{\ox n},\ldots,\rho_r
^{\ox n}$. Then we show that, on the one hand this removal will cause an error that matches the
right-hand side of equation~(\ref{eq:error-exponent}) in the exponent, and on the other hand the
pairwise overlaps between the supporting spaces of $\rho_1^{\ox n},\ldots,\rho_r^{\ox n}$ are made
sufficiently small, such that the Gram-Schmidt orthonormalization method is applicable. For the
sake of generality, we will actually realize these ideas for general nonnegative matrices $A_1,
\ldots,A_r$, yielding the following Theorem~\ref{thm:one-shot-achi}.

\begin{theorem}
  \label{thm:one-shot-achi}
  Let $A_1,\ldots,A_r\in\cP(\cH)$ be nonnegative matrices on a finite-dimensional Hilbert space $\cH$.
  For all $1\leq i \leq r$, let $A_i=\sum_{k=1}^{T_i}\lambda_{ik}Q_{ik}$ be the spectral
  decomposition of $A_i$, and write $T:=\max\{T_1,\ldots,T_r\}$. There exists a function $f(r,T)$ such
  that
  \begin{equation}
    \label{eq:one-shot-achi}
    P_e^*\left(\{A_1,\ldots,A_r\}\right)
       \leq f(r,T)\sum_{(i,j):i<j}\sum_{k,\ell} \min\{\lambda_{ik},\lambda_{j\ell}\} \tr Q_{ik}Q_{j\ell}
  \end{equation}
  and we have $f(r,T)<10(r\!-\!1)^2T^2$.
\end{theorem}

Our upper bound of equation~(\ref{eq:one-shot-achi}), up to an $r$- and $T$-dependent factor, coincides
with the multiple-state generalization of the lower bound of Nussbaum and Szko{\l}a~\cite{NussbaumSzkola06}.
To see this, using the result in~\cite{Qiu08}, we easily generalize the bound obtained
in~\cite{NussbaumSzkola06} and get
\begin{equation}
  \label{eq:one-shot-opti}
  P_e^*\left(\{A_1,\ldots,A_r\}\right) \geq
    \frac{1}{2(r-1)}\sum_{(i,j):i<j}\sum_{k,\ell} \min\{\lambda_{ik},\lambda_{j\ell}\} \tr Q_{ik}Q_{j\ell}.
\end{equation}
In the case that $r=2$, it is interesting to compare equation~(\ref{eq:one-shot-achi}) with the upper
bound of Audenaert et al~\cite{ACMMABV}:
\begin{equation}
  \label{Audenaert-achi}
  P_e^*\left(\{A_1,A_2\}\right) \leq \min_{0 \leq s \leq 1}\tr A_1^sA_2^{1-s}.
\end{equation}
While we see that our bound is stronger, in the sense that
\begin{equation}
  \label{eq:bounds-comparison}
  \sum_{k,\ell} \min\{\lambda_{1k},\lambda_{2\ell}\} \tr Q_{1k}Q_{2\ell}
                                                     \leq \min_{0 \leq s \leq 1}\tr A_1^sA_2^{1-s}
\end{equation}
is always true, we also notice that it is weaker because it has an additional multiplier depending on
the number of eigenspaces of the two states.

\section{Proofs}
  \label{section:proofs}
This section is dedicated to the proofs of Theorem~\ref{thm:quantum-multiple-Chernoff} and
Theorem~\ref{thm:one-shot-achi}. At first, we present a definition and some necessary lemmas in Section~\ref{subsection:technical-preliminaries}. Then we construct the measurement for discriminating
multiple quantum states in Section~\ref{subsection:measurements-construction}. Using this measurement,
we prove Theorem~\ref{thm:one-shot-achi} in section~\ref{subsection:proof of Theorem 2}. At last, built
on Theorem~\ref{thm:one-shot-achi}, Theorem~\ref{thm:quantum-multiple-Chernoff} will be proven in Section~\ref{subsection:proof of Theorem 1}.

\subsection{Technical preliminaries}
  \label{subsection:technical-preliminaries}
We begin with the definition of the operation ``$\epsilon$-subtraction'' between two projectors or
two subspaces. This operation, say, for two subspaces $S_1$ and $S_2$, reduces the overlap between
them by removing a subspace from $S_1$, actually, in the most efficient way. It will constitute a
key step in the construction of measurement.

\begin{definition}[$\epsilon$-subtraction]
  \label{def:epsilon-subtraction}
  Let $S_1, S_2$ be two subspaces of a Hilbert space $\cH$. Let $P_1, P_2\in\cP(\cH)$ be the
  projectors onto $S_1$ and $S_2$, respectively. Write $P_1P_2P_1$ in the spectral decomposition,
  $P_1P_2P_1=\sum_x\lambda_xQ_x$, with $Q_x$ being orthogonal projectors and $\sum_xQ_x=\1_{\cH}$.
  For $0\leq\epsilon\leq 1$, the $\epsilon$-subtraction of $P_2$ from $P_1$ is defined as
  \begin{equation}
    \label{eq:subtraction-projector}
    P_1\ominus_\epsilon P_2:=P_1-\sum_{x:\lambda_x\geq\epsilon^2, \lambda_x\neq 0}Q_x.
  \end{equation}
  Accordingly, the $\epsilon$-subtraction between subspaces is defined as
  \begin{equation}
    \label{eq:subtraction-space}
    S_1\ominus_\epsilon S_2:=\supp(P_1\ominus_\epsilon P_2).
  \end{equation}
\end{definition}

Note that in equation~(\ref{eq:subtraction-projector}) the constraint $\lambda_x\neq 0$ makes sense
only when $\epsilon=0$. The following lemma states some basic properties of the $\epsilon$-subtraction.

\begin{lemma}
  \label{lemma:epsilon-subtraction-prop}
  Let $S_1, S_2$ be two subspaces of a Hilbert space $\cH$. Let $P_1, P_2\in\cP(\cH)$ be the projectors
  onto $S_1$ and $S_2$, respectively. Write $S_1^\prime=S_1\ominus_\epsilon S_2$, and
  $P_1^\prime=P_1\ominus_\epsilon P_2$. Then
  \begin{enumerate}
    \item $S_1^\prime$ is a subspace of $S_1$; $P_1^\prime$ is a projector, and $0\leq P_1^\prime\leq P_1$.
    \item $S_1^\prime$ has bounded overlap with $S_2$:
          \[\max_{\ket{v_1}\in S_1^\prime,\ket{v_2}\in S_2} |\braket{v_1}{v_2}| \leq \epsilon,
          \]
          where the maximization is over unit vectors $\braket{v_1}{v_1}=\braket{v_2}{v_2}=1$.
    \item For $0<\epsilon\leq 1$, we have
          \[\tr(P_1-P_1^\prime)\leq \frac{1}{\epsilon^2}\tr P_1P_2.
          \]
  \end{enumerate}
\end{lemma}

\begin{proof}
Let $P_1P_2P_1=\sum_x\lambda_xQ_x$ be the spectral decomposition of $P_1P_2P_1$, with
$0\leq\lambda_x\leq 1$.

1. Obviously, $\supp(P_1P_2P_1)\subseteq S_1$. Thus the following three projectors satisfy
\begin{equation}
    \label{eq:proof-lemma1-a}
    \sum_{x:\lambda_x\geq\epsilon^2,\lambda_x\neq 0}Q_x \leq \sum_{x:\lambda_x\neq 0}Q_x \leq P_1.
  \end{equation}
This, together with Definition~\ref{def:epsilon-subtraction}, implies that $P_1^\prime$ is a
projector and satisfies $0\leq P_1^\prime \leq P_1$. The fact that $S_1^\prime$ is a subspace of
$S_1$, follows directly.

2. It follows from equation~(\ref{eq:proof-lemma1-a}) and Definition~\ref{def:epsilon-subtraction}
that we can write $S_1^\prime$ as
\[S_1^\prime = \left(\bigoplus_{x:0<\lambda_x<\epsilon^2}\supp(Q_x)\right)\bigoplus
                                               \left(\supp(Q_x)|_{\lambda_x= 0}\bigcap S_1\right).
\]
That is to say, $S_1^\prime$ is the direct sum of the eigenspaces of $P_1P_2P_1$ with corresponding
eigenvalues in the interval $(0, \epsilon^2)$, together with a subspace of the kernel of $P_1P_2P_1$.
So, for any unit vectors $\ket{v_1}\in S_1^\prime$, $\ket{v_2}\in S_2$,
\[|\braket{v_1}{v_2}|=\sqrt{\braket{v_1}{v_2}\braket{v_2}{v_1}}
  \leq \sqrt{\bra{v_1}P_2\ket{v_1}}=\sqrt{\bra{v_1}P_1P_2P_1\ket{v_1}} \leq \epsilon.
\]

3. This inequality can be verified as follows.
\[\begin{split}
  \tr(P_1-P_1^\prime)= &\tr\sum_{x:\lambda_x\geq\epsilon^2}Q_x  \\
                  \leq &\tr\sum_{x:\lambda_x\geq\epsilon^2}\frac{\lambda_x}{\epsilon^2}Q_x
                  \leq  \tr\sum_x \frac{\lambda_x}{\epsilon^2}Q_x  \\
                     = &\frac{1}{\epsilon^2}\tr P_1P_2P_1 = \frac{1}{\epsilon^2}\tr P_1P_2.
\end{split}\]
\end{proof}

Lemma~\ref{lemma:bounded-overlap} and Lemma~\ref{lemma:Projector-bound} below are basic results
for subspaces of an inner product space.
\begin{lemma}
  \label{lemma:bounded-overlap}
  Let $V, W$ be subspaces of a Hilbert space $\cH$, and have direct-sum decompositions
  $V=\bigoplus_{i=1}^{T_1}V_i$ and $W=\bigoplus_{j=1}^{T_2} W_j$. Suppose we have
  \[\max_{\ket{v}\in V_i,\ket{w}\in W_j}|\braket{v}{w}| \leq \delta,
    \quad \text{ for all } 1 \leq i\leq T_1,\, 1 \leq j\leq T_2.
  \]
  Then the overlap between $V$ and $W$ is bounded as
  \[\max_{\ket{v}\in V,\ket{w}\in W}|\braket{v}{w}| \leq \sqrt{T_1T_2}\delta.
  \]
\end{lemma}

\begin{proof}
Let $\ket{v}\in V$ and $\ket{w}\in W$ be any two unit vectors. Write $\ket{v}=\sum_i\alpha_i\ket{v_i}$
and $\ket{w}=\sum_j\beta_j\ket{w_j}$, with $\ket{v_i}\in V_i,\ket{w_j}\in W_j$ and $\sum_i|\alpha_i|^2
=\sum_j|\beta_j|^2=1$. Making use of the Cauchy-Schwarz inequality, we have
\[\begin{split}
  |\braket{v}{w}| =    &\left|\sum_{i,j}\overline{\alpha_i}\beta_j\braket{v_i}{w_j}\right| \\
                  \leq &\sum_{i,j}|\alpha_i|\cdot|\beta_j|\delta  \\
                  =    &\Big(\sum_{i=1}|\alpha_i|\Big)\Big(\sum_{j=1}|\beta_j|\Big)\delta \\
                  \leq &\sqrt{T_1T_2}\delta,
\end{split}\]
and this finishes the proof.
\end{proof}

\begin{lemma}
  \label{lemma:Projector-bound}
  Let $S_1, S_2,\ldots, S_r$ be subspaces of a Hilbert space, such that the overlaps between them are
  bounded as
  \[\max_{\ket{v_i}\in S_i,\ket{v_j}\in S_j}|\braket{v_i}{v_j}| \leq \delta,
    \quad  1 \leq i\neq j \leq r.
  \]
  For all $1 \leq i \leq r$, denote the projector onto $S_i$ as $P_i$, and denote the projector onto
  $S=S_1+S_2+\cdots+S_r$ as $P$. Suppose $\delta<\frac{1}{2(r-1)}$. Then
  \begin{equation}
    \label{eq:projector-bound}
    P \leq \frac{1-(r-1)\delta}{1-2(r-1)\delta}\sum_{i=1}^rP_i.
  \end{equation}
\end{lemma}

\begin{proof}
For an arbitrary unit vector $\ket{v}\in S$, write $\ket{v}=\sum_{i=1}^r\alpha_i\ket{v_i}$, with
$\ket{v_i}\in S_i$. Then
\[\begin{split}
  \bra{v}\sum_{k=1}^rP_k\ket{v}=&\sum_{(i,k):i=k}\sum_j \overline{\alpha_i}\alpha_j\bra{v_i}P_k\ket{v_j}
        +\sum_i\sum_{(j,k):j=k} \overline{\alpha_i}\alpha_j\bra{v_i}P_k\ket{v_j}  \\
   &\ \ -\sum_{\substack{(i,j,k):\\i=j=k}} \overline{\alpha_i}\alpha_j\bra{v_i}P_k\ket{v_j}
        +\sum_k\sum_{i:i\neq k}\sum_{j:j\neq k} \overline{\alpha_i}\alpha_j\bra{v_i}P_k\ket{v_j} \\
  =&2\braket{v}{v}-\sum_{i=1}^r|\alpha_i|^2+\sum_{k=1}^r\left(\bra{v}-\overline{\alpha_k}
                                             \bra{v_k}\right)P_k\left(\ket{v}-\alpha_k\ket{v_k}\right).
\end{split}\]
The last term is nonnegative. Besides, we have
\[\begin{split}
  1   =&\braket{v}{v}=\sum_{i=1}^r|\alpha_i|^2 +
                                \sum_{(i,j):i\neq j}\overline{\alpha_i}\alpha_j\braket{v_i}{v_j} \\
  \geq &\sum_{i=1}^r|\alpha_i|^2 -\sum_{(i,j):i\neq j}|\alpha_i|\cdot|\alpha_j|\cdot|\braket{v_i}{v_j}| \\
  \geq &\sum_{i=1}^r|\alpha_i|^2-\sum_{(i,j):i\neq j}\frac{\left(|\alpha_i|^2+|\alpha_j|^2\right)\delta}{2} \\
    =  &\big(1-(r-1)\delta\big)\sum_{i=1}^r|\alpha_i|^2.
\end{split}\]
Combining the above arguments, we get
\[
  \bra{v}\sum_{k=1}^rP_k\ket{v} \geq 2-\frac{1}{1-(r-1)\delta} = \frac{1-2(r-1)\delta}{1-(r-1)\delta},
\]
which implies equation~(\ref{eq:projector-bound}).
\end{proof}

\subsection{Construction of the measurements}
  \label{subsection:measurements-construction}
In the following, we will describe the procedure of constructing a family of projective measurements
$\{\Pi_1(\epsilon), \ldots, \Pi_{r-1}(\epsilon), \Pi_{r}(\epsilon)+\Pi_{r+1}(\epsilon)\}$, which will
be used to show that the right-hand side of equation~(\ref{eq:one-shot-achi}) is an achievable error
probability.

Our construction is similar to the ones explored in~\cite{NussbaumSzkola11a,NussbaumSzkola11} and
earlier in~\cite{Holevo78}, especially, in applying the Gram-Schmidt orthonormalization to states
that are ordered according to the corresponding eigenvalues, for formulating the projective
measurements. However, our method is also significantly different from those of~\cite{Holevo78,NussbaumSzkola11a,NussbaumSzkola11}. At first, instead of dealing with every
eigenvector of the hypothetic states individually, we treat each of the eigenspaces as a whole.
This, for i.i.d. states of the form $\omega^{\ox n}$, is reminiscent of the method of
types~\cite{Csiszar98,CsiszarKorner81}, from which we indeed benefit. Secondly, we carefully
remove from each of these eigenspaces a subspace, in a way such that the perturbation to the
hypothetic states is limited but the overlaps between the supporting spaces of them become
sufficiently small. As a result, we can effectively employ the Gram-Schmidt process to formulate
an approximately optimal measurement.

Recall that for $1\leq i\leq r$, $A_i=\sum_{k=1}^{T_i}\lambda_{ik}Q_{ik}$ is the spectral
decomposition. Let $S_{ik}:=\supp(Q_{ik})$ be the eigenspaces of $A_i$. From now on, we always identify
the subscript ``ik'' with ``(i,k)''. So $\lambda_{ik}$, $Q_{ik}$ and $S_{ik}$ are also denoted as
$\lambda_{(i,k)}$, $Q_{(i,k)}$ and $S_{(i,k)}$, respectively. Define the index set $\cO:=\bigcup_{i=1}^r
\{(i,k):k\in\mathbb{N}, 1\leq k\leq T_i\}$. Now we arrange all the eigenvalues
$\{\lambda_{ik}\}_{(i,k)\in\cO}$ in a decreasing order, and let $g:\{0,1,2,\ldots,|\cO|\}\mapsto\{(0,0)
\}\cup\cO$ be the bijection indicating the position of each $\lambda_{ik}$ in such an ordering:
\begin{equation}
  \label{eq:ordering}
\lambda_{g(1)} \geq \lambda_{g(2)} \geq \cdots \geq \lambda_{g(|\cO|)},
\end{equation}
and $g(0)=(0,0)$ is introduced for later convenience. Our construction consists of the following three
steps.

\textbf{Step 1}: reducing the overlaps between the eigenspaces. For this purpose, we employ the operation
$\epsilon$-subtraction to remove a subspace from each of these eigenspaces. Let $\ominus_\epsilon$ be
endowed with a left associativity, that is, $A\ominus_\epsilon B\ominus_\epsilon C:=(A\ominus_\epsilon B)
\ominus_\epsilon C$. Set $Q_{g(0)}=0$ and $S_{g(0)}=\{0\}$. We define
\begin{equation}
    \label{eq:dig-projectors}
    \tilde{Q}_{g(m)} := \begin{cases}
    & Q_{g(0)}, \quad  \text{if}\ m=0 \\
    & Q_{g(m)}\ominus_\epsilon Q_{g(0)}\ominus_\epsilon Q_{g(1)}\ominus_\epsilon \cdots \ominus_\epsilon Q_{g(m-1)}, \  \text{if}\ 1\leq m\leq |\cO|
    \end{cases}
\end{equation}
and
\begin{equation}
    \label{eq:dig-eigenspaces}
    \tilde{S}_{g(m)} := \begin{cases}
    & S_{g(0)}, \quad \text{if}\ m=0 \\
    & S_{g(m)}\ominus_\epsilon S_{g(0)}\ominus_\epsilon S_{g(1)}\ominus_\epsilon \cdots \ominus_\epsilon S_{g(m-1)}, \ \text{if}\ 1\leq m\leq |\cO|.
    \end{cases}
\end{equation}
Note that, according to Definition~\ref{def:epsilon-subtraction}, $\tilde{S}_{g(m)}=\supp\left(
\tilde{Q}_{g(m)}\right)$. Now we denote
\begin{equation}
  \label{eq:post-digging-states}
\tilde{A}_i:=\sum_{k=1}^{T_i}\lambda_{ik}\tilde{Q}_{ik}, \quad 1\leq i\leq r.
\end{equation}
We will show later that, for the purpose of the current paper, $\tilde{A}_i$ is a good approximation
of $A_i$.

\textbf{Step 2}: orthogonalizing the eigenspaces. To formulate the projective measurement, we need
to assign each of the states $\{A_i\}_i$ an orthogonal subspace, for the projectors to be supported
on. To do so, we treat $\tilde{A}_i\text{'s}$ as representives of $A_i\text{'s}$, and orthogonalize
the subspaces $\{\tilde{S}_{g(m)}\}_m$, using a Gram-Schmidt process. Define $\hat{S}_{g(0)}:=\{0\}$,
and for all $1\leq m\leq |\cO|$,
\begin{equation}
  \label{eq:orthogonal-subspaces-2}
\hat{S}_{g(m)}:=\left(\tilde{S}_{g(0)}+\cdots +\tilde{S}_{g(m)}\right)\ominus_1
                                 \left(\tilde{S}_{g(0)}+\cdots +\tilde{S}_{g(m-1)}\right),
\end{equation}
where $\ominus_1$ is the operation ``$\epsilon$-subtraction'' with $\epsilon=1$. Recalling Definition~\ref{def:epsilon-subtraction}, we easily see that the right-hand side of
equation~(\ref{eq:orthogonal-subspaces-2}) is just the orthogonal complement of $\tilde{S}_{g(0)}
+\cdots +\tilde{S}_{g(m-1)}$, in the space $\tilde{S}_{g(0)}+\cdots +\tilde{S}_{g(m)}$, noticing
that obviously the former is a subspace of the latter. So the subspaces $\hat{S}_{g(1)},\ldots,\hat{S}
_{g(m)}$ are mutually orthogonal. Thus the definition in equation~(\ref{eq:orthogonal-subspaces-2})
is equivalent to
\begin{equation}
  \label{eq:orthogonal-subspaces-3}
  \tilde{S}_{g(1)}+\cdots +\tilde{S}_{g(m)}=\bigoplus_{t=1}^m \hat{S}_{g(t)}, \quad \text{for all }
                                                                                \ 1\leq m\leq|\cO|.
\end{equation}
Note that it is possible that $\hat{S}_{g(m)}=\{0\}$ for certain values of $m$, and in these cases
we have $\Proj\left(\hat{S}_{g(m)}\right)=0$.

\textbf{Step 3}: defining the family of projective measurements. We set
\begin{equation}
  \label{eq:measurement-projectors-1}
\Pi_i(\epsilon):=\sum_{k=1}^{T_i} \Proj\left(\hat{S}_{ik}\right), \quad 1\leq i\leq r,
\end{equation}
and let
\begin{equation}
  \label{eq:measurement-projectors-2}
\Pi_{r+1}(\epsilon):=\1-\sum_{i=1}^{r} \Pi_i(\epsilon).
\end{equation}
Here the parameter $\epsilon$ is introduced in step 1. By definition, $\Pi_1(\epsilon),\ldots,
\Pi_{r+1}(\epsilon)$ are orthogonal projectors and $\sum_{i=1}^{r+1}\Pi_i(\epsilon)=\1$. So, they form
a projective measurement. Our strategy for testing $A_1,\ldots,A_r$ is that, if we get the measurement
outcome associated with $\Pi_i(\epsilon)$, we conclude that the state is $A_i$. For the outcome
associated with the extra projector $\Pi_{r+1}(\epsilon)$, we can make any decision, or just report an
error; here we simply assign it to $A_r$. Thus, the family of measurements that we construct for use is
\begin{equation}
  \label{eq:measurement}
\Pi=\left\{\Pi_1(\epsilon),\ldots,\Pi_{r-1}(\epsilon), \Pi_{r}(\epsilon)+\Pi_{r+1}(\epsilon)\right\}.
\end{equation}

\subsection{Proof of the one-shot achievability bound: Theorem~\ref{thm:one-shot-achi}}
  \label{subsection:proof of Theorem 2}
We show that, for properly chosen $\epsilon\in[0,1]$, the measurement constructed in
Section~\ref{subsection:measurements-construction} will achieve an error probability that equals
the right-hand side of equation~(\ref{eq:one-shot-achi}).

\begin{proof}[Proof of Theorem~\ref{thm:one-shot-achi}]
For the ensemble of nonnegative matrices $\cA=\{A_1,\dots, A_r\}$, and the measurement $\Pi$
specified in equation~(\ref{eq:measurement}), we have
\[\begin{split}
  P_e(\cA;\Pi)=&\sum_{i=1}^{r-1} \tr A_i \left(\1-\Pi_i(\epsilon)\right)+\tr
                                   A_r\left(\1-\Pi_{r}(\epsilon)-\Pi_{r+1}(\epsilon)\right) \\
          \leq &\sum_{i=1}^{r} \tr A_i \left(\1-\Pi_i(\epsilon)\right).
\end{split}\]
We now make use of the matrices $\{\tilde{A}_i\}$, which are defined in step 1 of the measurement construction
in section~\ref{subsection:measurements-construction}; cf. equantion~(\ref{eq:post-digging-states}).
Substituting $(A_i-\tilde{A}_i)+\tilde{A}_i$ for $A_i$, and noticing that it is an obvious result of
equations~(\ref{eq:dig-projectors}) and (\ref{eq:post-digging-states}) that $A_i-\tilde{A}_i\geq 0$,
we further bound the error probability as
\begin{equation}
  \label{eq:error-bound}
  P_e(\cA;\Pi)\leq \sum_{i=1}^{r} \tr (A_i-\tilde{A}_i)
                           + \sum_{i=1}^{r} \tr \tilde{A}_i \left(\1-\Pi_i(\epsilon)\right).
\end{equation}
In the following, we will evaluate the two terms of the right-hand side of equation~(\ref{eq:error-bound}),
separately.

Invoking the spectral decomposition of $A_i$, and using equation~(\ref{eq:post-digging-states}),
we can write
\begin{equation}\begin{split}
  \label{eq:first-term-bound-1}
    \sum_{i=1}^{r}\tr(A_i-\tilde{A}_i)
  =&\sum_{i=1}^{r}\sum_{k=1}^{T_i}\lambda_{ik}\tr(Q_{ik}-\tilde{Q}_{ik})\\
  =&\sum_{m=1}^{|\cO|}\lambda_{g(m)}\tr\left(Q_{g(m)}-\tilde{Q}_{g(m)}\right),
\end{split}\end{equation}
where we have used the map $g$, introduced in the previous section, to indicate the subscripts.
For each integer $2\leq m\leq |\cO|$, applying the third result of Lemma~\ref{lemma:epsilon-subtraction-prop}
to the $\epsilon$-subtraction formulas
\[
  \left(Q_{g(m)}\ominus_\epsilon Q_{g(0)}\ominus_\epsilon Q_{g(1)}\ominus_\epsilon \cdots \ominus_\epsilon Q_{g(t-1)}\right)\ominus_\epsilon Q_{g(t)}\, ,   \quad 1\leq t\leq m-1,
\]
we obtain
\begin{equation}\begin{split}
  \label{eq:first-term-bound-2}
  &\tr \left(Q_{g(m)}\ominus_\epsilon Q_{g(0)}\ominus_\epsilon \cdots \ominus_\epsilon Q_{g(t-1)}\right)
   - \tr \left(Q_{g(m)}\ominus_\epsilon Q_{g(0)}\ominus_\epsilon \cdots \ominus_\epsilon Q_{g(t)}\right) \\
  &\leq \frac{1}{\epsilon^2}\tr\left(Q_{g(m)}\ominus_\epsilon Q_{g(0)}\ominus_\epsilon \cdots
                                                              \ominus_\epsilon Q_{g(t-1)}\right)Q_{g(t)} \\
  &\leq \frac{1}{\epsilon^2}\tr Q_{g(m)}Q_{g(t)},
\end{split}\end{equation}
where for the last inequality we have used repeatedly the first result of Lemma~\ref{lemma:epsilon-subtraction-prop}.
Summing over $t\in\{1,2,\ldots,m-1\}$ at both the first and the last line of
Equation~(\ref{eq:first-term-bound-2}), yields
\begin{equation}\begin{split}
  \label{eq:first-term-bound-3}
       &\tr Q_{g(m)} - \tr \left(Q_{g(m)}\ominus_\epsilon Q_{g(0)}\ominus_\epsilon Q_{g(1)}
                                               \ominus_\epsilon\cdots \ominus_\epsilon Q_{g(m-1)}\right) \\
  \leq & \frac{1}{\epsilon^2}\sum_{t=1}^{m-1}\tr Q_{g(m)}Q_{g(t)},
\end{split}\end{equation}
for all $2\leq m\leq |\cO|$. Combining equations (\ref{eq:dig-projectors}) and (\ref{eq:first-term-bound-3}),
and directly verifying the case $m=1$, we get
\begin{equation}
  \label{eq:first-term-bound-4}
       \tr \left(Q_{g(m)}-\tilde{Q}_{g(m)}\right)
  \leq \frac{1}{\epsilon^2}\sum_{t=0}^{m-1}\tr Q_{g(m)}Q_{g(t)}, \quad  1\leq m\leq |\cO|.
\end{equation}
Eventually, inserting equation~(\ref{eq:first-term-bound-4}) into equation~(\ref{eq:first-term-bound-1}),
and changing the subscripts, we arrive at
\begin{equation}\begin{split}
  \label{eq:first-term-bound}
  \sum_{i=1}^{r}\tr(A_i-\tilde{A}_i)
  \leq & \frac{1}{\epsilon^2}\sum_{m=1}^{|\cO|}\sum_{t=0}^{m-1}\lambda_{g(m)}\tr Q_{g(m)}Q_{g(t)} \\
  =&\frac{1}{\epsilon^2}\sum_{(i,j):i<j}\sum_{k,\ell}\min\{\lambda_{ik},\lambda_{j\ell}\} \tr Q_{ik}Q_{j\ell}.
\end{split}\end{equation}

Now we evaluate the second term of the right-hand side of equation~(\ref{eq:error-bound}). Using
equations~(\ref{eq:post-digging-states}) and (\ref{eq:measurement-projectors-1}), and employing the
map $g$ to indicate the subscripts, we can write
\begin{equation}\begin{split}
  \label{eq:second-term-bound-1}
     \sum_{i=1}^{r} \tr \tilde{A}_i \left(\1-\Pi_i(\epsilon)\right)
  =    &\sum_{i=1}^{r} \sum_{k=1}^{T_i}\lambda_{ik}\tr\tilde{Q}_{ik}\left(\1-
                                 \sum_{\ell=1}^{T_i}\Proj\left(\hat{S}_{i\ell}\right)\right) \\
  \leq &\sum_{i=1}^{r} \sum_{k=1}^{T_i}\lambda_{ik}\tr\tilde{Q}_{ik}\left(\1-
                                                           \Proj\left(\hat{S}_{ik}\right)\right)  \\
  =    &\sum_{m=1}^{|\cO|} \lambda_{g(m)}\tr\tilde{Q}_{g(m)}\left(\1-
                                                           \Proj\left(\hat{S}_{g(m)}\right)\right).
\end{split}\end{equation}
Equation~(\ref{eq:orthogonal-subspaces-3}) implies that
\[
  \supp\left(\tilde{Q}_{g(m)}\right)=\tilde{S}_{g(m)}\subseteq \bigoplus_{t=1}^m\hat{S}_{g(t)}.
\]
As a result, the identity matrix in the third line of equation~(\ref{eq:second-term-bound-1})
can be replaced by $\sum_{t=1}^m \Proj\left(\hat{S}_{g(t)}\right)$. This gives
\begin{equation}\begin{split}
  \label{eq:second-term-bound-2}
     \sum_{i=1}^{r} \tr \tilde{A}_i \left(\1-\Pi_i(\epsilon)\right)
  \leq &\sum_{m=1}^{|\cO|} \lambda_{g(m)}\tr\tilde{Q}_{g(m)}\sum_{t=0}^{m-1}\Proj\left(\hat{S}_{g(t)}\right) \\
  =    &\sum_{m=1}^{|\cO|} \lambda_{g(m)}\tr\tilde{Q}_{g(m)}\,\Proj\left(\sum_{t=0}^{m-1}\tilde{S}_{g(t)}\right),
\end{split}\end{equation}
where we have $\hat{S}_{g(0)}=\tilde{S}_{g(0)}=\{0\}$, and for the equality we used
equation~(\ref{eq:orthogonal-subspaces-3}). The next step is to upper bound $\Proj\left(\sum_{t=0}
^{m-1}\tilde{S}_{g(t)}\right)$, with a quantity in terms of $\sum_{t=0}^{m-1}\tilde{Q}_{g(t)}$.
This can be done by directly applying Lemma~\ref{lemma:Projector-bound}. However, we notice that,
for each $1\leq i\leq r$ the subspaces in the set $\{\tilde{S}_{ik}\}_k$ are orthogonal and we can
make use of this fact to derive a tighter bound. For a pair of numbers $(x,y)$, let $[(x,y)]_1$ denote
the first component: $[(x,y)]_1=x$. We write
\begin{equation}
  \label{eq:second-term-bound-3}
  \sum_{t=0}^{m-1}\tilde{S}_{g(t)}=\sum_{i=1}^r \tilde{S}_i^m,\quad\text{with}\ \tilde{S}_i^m :=
                          \bigoplus_{\substack{t:0\leq t\leq m-1,\\ [g(t)]_1=i}}\tilde{S}_{g(t)}.
\end{equation}
We will use Lemma~\ref{lemma:bounded-overlap} to bound the overlaps between each pair of the
subspaces $\{\tilde{S}_i^m\}_{i=1}^r$, and then we apply Lemma~\ref{lemma:Projector-bound}.
Although we will only get a slightly better bound, compared to applying Lemma~\ref{lemma:Projector-bound}
directly, it is possible to make this improvement bigger by strengthening the result of
Lemma~\ref{lemma:Projector-bound}. Now, due to Lemma~\ref{lemma:epsilon-subtraction-prop}
and the definition of $\tilde{S}_{g(t)}$ [cf. equation~(\ref{eq:dig-eigenspaces})], we can bound
\[
  \max_{\ket{v}\in\tilde{S}_{g(t)},\ket{v'}\in S_{g(t')}}\left|\braket{v}{v'}\right|\leq\epsilon,
  \quad 0\leq t'< t\leq|\cO|.
\]
Since $\tilde{S}_{g(t')}$ is a subspace of $S_{g(t')}$, it follows that
\[
  \max_{\ket{v}\in\tilde{S}_{g(t)},\ket{v'}\in\tilde{S}_{g(t')}}\left|\braket{v}{v'}\right|\leq\epsilon,
  \quad 0\leq t\neq t'\leq|\cO|.
\]
Recalling the spectral decomposition of each $A_i$, we notice that the direct sum in
equation~(\ref{eq:second-term-bound-3}) has at most $T_i$ terms. So, an application of
Lemma~\ref{lemma:bounded-overlap} gives us that, for every $1\leq m\leq|\cO|$,
\begin{equation}
  \label{eq:second-term-bound-5}
  \max_{\ket{v}\in\tilde{S}_i^m,\ket{w}\in\tilde{S}_j^m}\left|\braket{v}{w}\right|\leq T\epsilon,
  \quad 1\leq i\neq j\leq r,
\end{equation}
where $T=\max\{T_1,\ldots,T_r\}$. Equation~(\ref{eq:second-term-bound-5}) allows us to apply
Lemma~\ref{lemma:Projector-bound} and obtain
\begin{equation}
  \label{eq:second-term-bound-7}
  \Proj\left(\sum_{t=0}^{m-1}\tilde{S}_{g(t)}\right)
       \leq \frac{1\!-\!(r\!-\!1)T\epsilon}{1\!-\!2(r\!-\!1)T\epsilon}
                   \sum_{t=0}^{m-1}\tilde{Q}_{g(t)}\,, \quad 1\leq m\leq|\cO|,
\end{equation}
for which we have also used equation~(\ref{eq:second-term-bound-3}) and the fact that $\Proj\left(
\tilde{S}_{g(t)}\right)=\tilde{Q}_{g(t)}$. Note that here the condition on $\delta$ in
Lemma~\ref{lemma:Projector-bound} is satisfied for our later choice of $\epsilon$. Now, inserting equation~(\ref{eq:second-term-bound-7}) into equation~(\ref{eq:second-term-bound-2}), and making
use of the relation $\tilde{Q}_{g(m)}\leq Q_{g(m)}$ for all $0\leq m\leq|\cO|$, we arrive at
\begin{equation}
  \label{eq:second-term-bound-8}
  \sum_{i=1}^{r} \tr \tilde{A}_i \left(\1-\Pi_i(\epsilon)\right)
         \leq \frac{1\!-\!(r\!-\!1)T\epsilon}{1\!-\!2(r\!-\!1)T\epsilon}
                \sum_{m=1}^{|\cO|}\lambda_{g(m)}\sum_{t=0}^{m-1}\tr Q_{g(m)} Q_{g(t)},
\end{equation}
which translates to
\begin{equation}
  \label{eq:second-term-bound}
  \sum_{i=1}^{r} \tr \tilde{A}_i \left(\1-\Pi_i(\epsilon)\right) \leq \frac{1\!-\!(r\!-\!1)T\epsilon}{1\!-\!2(r\!-\!1)T\epsilon}
       \sum_{(i,j):i<j}\sum_{k,\ell} \min\{\lambda_{ik},\lambda_{j\ell}\} \tr Q_{ik}Q_{j\ell}.
\end{equation}

Eventually, inserting equations~(\ref{eq:first-term-bound}) and~(\ref{eq:second-term-bound}) into equation~(\ref{eq:error-bound})and setting $\epsilon=\frac{2}{5(r-1)T}$ lets us obtain
equation~(\ref{eq:one-shot-achi}), with
\[
  f(r,T)=\frac{25(r-1)^2T^2}{4}+3 < 10(r\!-\!1)^2T^2,
\]
and we are done.
\end{proof}

\subsection{Proof of the error exponent}
  \label{subsection:proof of Theorem 1}
We are now ready for the proof of Theorem~\ref{thm:quantum-multiple-Chernoff}.

\begin{proof}[Proof of Theorem~\ref{thm:quantum-multiple-Chernoff}]
For the achievability part, we use Theorem~\ref{thm:one-shot-achi}. Let $d=|\cH|$ be the dimension
of the associated Hilbert space of the states $\rho_1,\ldots, \rho_r$. The type counting lemma
(see, e.g.~\cite{CT91}, Theorem 12.1.1) provides that the number of eigenspaces of the states
$\rho_1^{\ox n},\ldots, \rho_r^{\ox n}$ satisfies
\[
  \Omega\left(\rho_i^{\ox n}\right) \leq (n+1)^d, \quad \forall\ 1\leq i\leq r.
\]
For all $1\leq i\leq r$, let $\rho_i^{\ox n}=\sum_k\lambda_{ik}^{(n)}Q_{ik}^{(n)}$ be written in the
spectral decomposition. Theorem~\ref{thm:one-shot-achi} gives
\begin{equation}\begin{split}
  \label{eq:theorem-1-proof-1}
       &P_e^*\left(\{p_1\rho_1^{\ox n},\ldots,p_r\rho_r^{\ox n}\}\right)  \\
  \leq &10(r\!-\!1)^2(n+1)^{2d} \sum_{(i,j):i<j}\sum_{k,\ell} \min\{p_i\lambda_{ik}^{(n)},
                                   p_j\lambda_{j\ell}^{(n)}\} \tr Q_{ik}^{(n)}Q_{j\ell}^{(n)}.
\end{split}\end{equation}
Furthermore, for any $1\leq i<j\leq r$, we have
\begin{equation}\begin{split}
  \label{eq:theorem-1-proof-2}
       &\sum_{k,\ell} \min\{p_i\lambda_{ik}^{(n)},p_j\lambda_{j\ell}^{(n)}\} \tr Q_{ik}^{(n)}Q_{j\ell}^{(n)} \\
  \leq &\max\{p_i,p_j\}\min_{0\leq s\leq 1} \sum_{k,\ell} \big(\lambda_{ik}^{(n)}\big)^s
                               \big(\lambda_{j\ell}^{(n)}\big)^{1-s} \tr Q_{ik}^{(n)}Q_{j\ell}^{(n)} \\
   =   &\max\{p_i,p_j\}\min_{0\leq s\leq 1} \left(\tr\rho_i^s\rho_j^{1-s}\right)^n.
\end{split}\end{equation}
Inserting equation~(\ref{eq:theorem-1-proof-2}) into equation~(\ref{eq:theorem-1-proof-1}), together
with some basic calculus, results in
\begin{equation}\begin{split}
  \label{eq:theorem-1-proof-3}
       &P_e^*\left(\{p_1\rho_1^{\ox n},\ldots,p_r\rho_r^{\ox n}\}\right)  \\
  \leq &10(r\!-\!1)^2C_r^2(n+1)^{2d}\max\{p_1,\ldots,p_r\} \max_{(i,j):i\neq j}\min_{0\leq s\leq 1}
                                                                 \left(\tr\rho_i^s\rho_j^{1-s}\right)^n,
\end{split}\end{equation}
Where $C_r^2=\frac{r(r-1)}{2}$ is a binomial coefficient. From equation~(\ref{eq:theorem-1-proof-3})
we easily derive
\begin{equation}\begin{split}
  \label{eq:theorem-1-proof-4}
       \liminf_{n\rar\infty}\frac{-1}{n}\log P_e^*\left(\{p_1\rho_1^{\ox n},\ldots,p_r\rho_r^{\ox n}\}\right)
  \geq \min_{(i,j):i\neq j}\max_{0\leq s\leq 1}\left\{-\log\tr\rho_i^s\rho_j^{1-s}\right\}.
\end{split}\end{equation}

On the other hand, the optimality part, that
\begin{equation}\begin{split}
  \label{eq:theorem-1-proof-5}
       \limsup_{n\rar\infty}\frac{-1}{n}\log P_e^*\left(\{p_1\rho_1^{\ox n},\ldots,p_r\rho_r^{\ox n}\}\right)
  \leq \min_{(i,j):i\neq j}\max_{0\leq s\leq 1}\left\{-\log\tr\rho_i^s\rho_j^{1-s}\right\},
\end{split}\end{equation}
is a straightforward generalization of the $r=2$ situation~\cite{NussbaumSzkola06}; see~\cite{NussbaumSzkola11a}
for the proof. Alternatively, one can start with the one-shot bound of equation~(\ref{eq:one-shot-opti}).
Then we use the fact that equation~(\ref{eq:bounds-comparison}), when applied to the i.i.d. states
and acted by ``$\frac{-1}{n}\log$'' at both sides, becomes asymptotically an equality. Note that
this is still based on the results of Nussbaum and Szko{\l}a in~\cite{NussbaumSzkola06}.

At last, equation~(\ref{eq:theorem-1-proof-4}) and equation~(\ref{eq:theorem-1-proof-5}) together
are obviously equivalent to equation~(\ref{eq:error-exponent}) and we conclude the proof of Theorem~\ref{thm:quantum-multiple-Chernoff}.
\end{proof}

\section{Discussion}
  \label{section:discussions}
By explicitly constructing a family of asymptotically optimal measurements for testing quantum
hypotheses $\{\rho_1^{\ox n},\ldots,\rho_r^{\ox n}\}$, we have proven the achievability of the
multiple quantum Chernoff distance, and eventually established that this is the optimal rate
exponent at which the error decays.

In the nonasymptotic setting, we have obtained a new upper bound for the optimal average error
probability in discriminating a set of density matrices $\{A_1,\ldots, A_r\}$, which satisfy
$A_i\geq 0$ and are not necessarily normalized. Yuen, Kennedy and Lax~\cite{YKL75} derived a
formula for the optimal average error:
\begin{equation}
  \label{eq:YKL}
  P_e^*(A_1,\ldots,A_r)=\tr\sum_iA_i-\min\left\{\tr X:X\geq A_i,\ i=1,\ldots,r\right\};
\end{equation}
see also~\cite{KRS09} and \cite{Aud-Mosonyi14} for alternative formulations. However, the fact
that equation~(\ref{eq:YKL}) involves an optimization problem itself, makes it difficult to
apply this formula directly. Our upper bound stated in Theorem~\ref{thm:one-shot-achi}, though
looser compared to equation~(\ref{eq:YKL}), has an explicit form and there is a dual lower bound
as shown in equation~(\ref{eq:one-shot-opti}). We thus hope that it will find more applications.

We wonder whether the states-dependent factor $f(r,T)$ can be replaced by a constant, or
at least can be improved such that it only depends on $r$ (see also a similar conjecture made
in~\cite{Aud-Mosonyi14}). While it is possible that we can improve Lemma~\ref{lemma:Projector-bound}
to give a better bound on $f(r,T)$, we do not think that this can remove the dependence on $T$
and $r$. In this direction, the pretty good measurement~\cite{BarnumKnill02,HausladenWootters94}
and its variant~\cite{Tyson09}, both of which achieve an error probability lying between $P_e^*$ and
$2P_e^*$, may be useful tools to try. In fact, in Theorem~\ref{thm:one-shot-achi} the dependence
of our bound on $T$ is not necessary: using the argument in~\cite{TomamichelHayashi12}, we can
convert it into a dependence on the relation between the maximal and the minimal eigenvalues of the
hypothetic states; see Proposition~\ref{prop:factor-dep-alt} below and the proof in the Appendix.
This conversion is useful when the spectrum of each $A_i$ is sufficiently flat, no matter how big
the number of their eigenspaces is.
\begin{proposition}
  \label{prop:factor-dep-alt}
For all $i=1,\ldots,r$, let $\lambda_{max}(A_i)$ be the maximal eigenvalue, and $\lambda_{min}(A_i)$
be the minimal nonzero eigenvalue of $A_i$. Denote
\[
  L:=\max\left\{\left\lfloor\log_2\frac{2\lambda_{max}(A_1)}{\lambda_{min}(A_1)}\right\rfloor,\ldots,
           \left\lfloor\log_2\frac{2\lambda_{max}(A_r)}{\lambda_{min}(A_r)}\right\rfloor\right\}.
\]
Then, in Theorem~\ref{thm:one-shot-achi} the states-dependent factor $f(r,T)$ can be replaced by
$h(r,L):=40(r-1)^2L^2$.
\end{proposition}

Another interesting question is how our method can be extended to deal with the problem of
discriminating correlated states, where each of the hypothetic states $\rho_1^{(n)},\ldots,\rho_r^{(n)}$
can be correlated among the $n$ subsystems. The upper bound stated in Theorem~\ref{thm:one-shot-achi}
(also in Proposition~\ref{prop:factor-dep-alt} for an alternative states-dependent factor), together
with the dual lower bound of equation~(\ref{eq:one-shot-opti}), can be used to analyse the asymptotic
behavior of the error. This method may identify the optimal error exponent which can be quite different
from reasonable generalizations of the Chernoff distance, in contrast to previous works which under
certain conditions yield the mean quantum Chernoff distance; see, for example,~\cite{HMO07, Mosonyi09, MHOF08,NussbaumSzkola10}. However, the main difficulty we will confront in this method is to characterize
the spectral decomposition of the correlated states when $n$ goes to infinity. At last, a
particularly interesting problem in this setting, proposed by Audenaert and Mosonyi~\cite{Aud-Mosonyi14},
is testing composite hypotheses, say, $\rho^{\otimes n}$ versus $\sum_iq_i\sigma_i^{\otimes n}$. Here
the sum may be replaced by an integral. See also discussions in~\cite{BDKSSS05} and ~\cite{BHOS15}
of this problem in the asymmetric case of Stein's lemma. While our method for proving Theorem~\ref{thm:quantum-multiple-Chernoff} does shed some light on this problem, it seems that a
complete solution needs further ideas.

\medskip
{\bf Acknowledgments.} The author is grateful to Charles Bennett, Aram Harrow, Graeme Smith, John Smolin
and Andreas Winter, for very helpful discussions and comments. He further thanks Fernando Brand\~{a}o,
Mil\'an Mosonyi, Marco Tomamichel and the anonymous referees for their enthusiasm and/or helpful comments.

\bigskip
\appendix
\title{Appendix A: proof of Proposition~\ref{prop:factor-dep-alt}}
\begin{proof}
For an arbitrary nonnegative matrix $A=\sum_k\lambda_kQ_k$ written in the spectral decomposition
form, define the modified version of $A$ as
\[
  A^\prime=\sum_{m=1}^M 2^m\lambda_{min}(A) \sum_{k:\lambda_k\in\cO_m}Q_k,
\]
where $M:=\left\lfloor\log_2\frac{2\lambda_{max}(A)}{\lambda_{min}(A)}\right\rfloor=\Omega(A^\prime)$
and $\cO_m:=\{\lambda_k:2^{m-1}\lambda_{min}(A)\leq\lambda_k<2^{m}\lambda_{min}(A)\}$. Then we
have $A\leq A^\prime \leq 2A$, and also $A$ and $A'$ commute. Now for $A_1,\ldots,A_r$, we define
$A_1',\ldots,A_r'$ in a similarly way as $A'$ was defined. Obviously, $\Omega(A_i^\prime)=
\left\lfloor\log_2\frac{2\lambda_{max}(A_i)}{\lambda_{min}(A_i)}\right\rfloor$. Applying
Theorem~\ref{thm:one-shot-achi}, we can evaluate
\begin{equation}\begin{split}
  \label{eq:prop-proof-1}
       & P_e^*\left(\{A_1',\ldots,A_r'\}\right) \\
  \leq & 10(r-1)^2L^2\cdot 4\sum_{(i,j):i<j}\sum_{k,\ell} \min\{\lambda_{ik},\lambda_{j\ell}\}
                                                                             \tr Q_{ik}Q_{j\ell}.
\end{split}\end{equation}
On the other hand, since for all $i$, $A_i\leq A_i'$, we have by the definition of $P_e^*$ that
\begin{equation}
  \label{eq:prop-proof-2}
  P_e^*\left(\{A_1,\ldots,A_r\}\right) \leq P_e^*\left(\{A_1',\ldots,A_r'\}\right).
\end{equation}
Equations~(\ref{eq:prop-proof-1}) and (\ref{eq:prop-proof-2}) together lead to the advertised
result.
\end{proof}


\begin{thebibliography}{99}

\bibitem{ACMMABV} \textsc{Audenaert, K. M. R., Casamiglia, J., Munoz-Tapia, R., Bagan, E.,
Masanes, Ll., Acin, A.} and \textsc{Verstraete, F.} (2007).
Discriminating states: the quantum Chernoff bound.
\textit{Phys. Rev. Lett.} \textbf{98} 160501. arXiv:quant-ph/0610027.

\bibitem{Aud-Mosonyi14} \textsc{Audenaert, K. M. R.} and \textsc{Mosonyi, M.} (2014).
Upper bounds on the error probabilities and asymptotic error exponents in quantum multiple state discrimination.
\textit{J. Math. Phys.} \textbf{55} 102201.

\bibitem{ANSV08}
\textsc{Audenaert, K. M. R., Nussbaum, M., Szko\l a, A.} and \textsc{Verstraete, F.} (2008).
Asymptotic error rates in quantum hypothesis testing.
\textit{Comm. Math. Phys.} \textbf{279} (1) 251--283.

\bibitem{BarnumKnill02} \textsc{Barnum, H.} and \textsc{Knill, E.} (2002).
Reversing quantum dynamics with near-optimal quantum and classical fidelity.
\textit{J. Math. Phys.} \textbf{43} 2097--2106.

\bibitem{BDKSSS05} \textsc{Bjelakovi\'c, I., Deuschel, J. D., Kr\"uger, T., Seiler, R.,
Siegmund-Schultze, Ra.} and \textsc{Szko{\l}a, A.} (2005).
A quantum version of Sanov's theorem.
\textit{Comm. Math. Phys.} \textbf{260} (3) 659--671.

\bibitem{Blahut74} \textsc{Blahut, R. E.} (1974).
Hypothesis testing and information theory.
\textit{IEEE Trans. Inf. Theory} \textbf{20} (4) 405--417.

\bibitem{BHOS15} \textsc{Brand\~ao, F. G. S. L., Harrow, A. W., Oppenheim, J.} and \textsc{Strelchuk, S.} (2015).
Quantum conditional mutual information, reconstructed states, and state redistribution.
\textit{Phys. Rev. Lett.} \textbf{115} 050501.

\bibitem{BrandaoPlenio09} \textsc{Brand\~ao, F. G. S. L.} and  \textsc{Plenio, M. B.} (2010).
A generalization of quantum Stein's lemma.
\textit{Comm. Math. Phys.}  \textbf{295} (3) 791--828.

\bibitem{Chernoff52} \textsc{Chernoff, H.} (1952).
A measure of asymptotic efficiency for tests of a hypothesis based on the sum of observations.
\textit{Ann. Math. Statist.} \textbf{23} (4) 493--507.

\bibitem {CT91} \textsc{Cover, T. M.} and \textsc{ Thomas, J. A.} (1991).
\textit{Elements of Information Theory}.
Wiley Series in Telecommunications, John Wiley \& Sons, New York.

\bibitem{Csiszar98} \textsc{Csisz\'ar, I.} (1998).
The method of types.
\textit{IEEE Trans. Inf. Theory} \textbf{44} (6) 2505--2523.

\bibitem{CsiszarKorner81} \textsc{Csisz\'ar, I.} and \textsc{K\"orner, J.} (1981).
\textit{Information Theory: Coding Theorems for Discrete Memoryless Systems}.
Academic Press, New York.

\bibitem{CsiszarLongo71} \textsc{Csisz\'ar, I.} and \textsc{Longo, G.} (1971).
On the error exponent for source coding and for testing simple statistical hypotheses.
\textit{Studia Sci. Math. Hungarica} \textbf{6} 181--191.

\bibitem{HanKobayashi89} \textsc{Han, T. S.} and \textsc{Kobayashi, K.} (1989).
The strong converse theorem for hypothesis testing.
\textit{IEEE Trans. Inf. Theory} \textbf{35} (1) 178--180.

\bibitem{HausladenWootters94} \textsc{Hausladen, P.} and \textsc{Wootters, W.} (1994).
A ``pretty good'' measurement for disdinguishing quantum states.
\textit{J. Mod. Opt.} \textbf{41} 2385.

\bibitem{Helstrom76} \textsc{Helstrom, C. W.} (1976).
\textit{Quantum Detection and Estimation Theory}.
Academic Press, New York.

\bibitem{HMO07} \textsc{Hiai, F., Mosonyi, M.} and \textsc{Ogawa, T.} (2007).
Large deviations and Chernoff bound for certain correlated states on a spin chain.
\textit{J. Math. Phys.} \textbf{48} 123301.

\bibitem{HiaiPetz91} \textsc{Hiai, F.} and \textsc{Petz, D.} (1991).
The proper formula for relative entropy and its asymptotics in quantum probability.
\textit{Comm. Math. Phys.} \textbf{143} (1) 99--114.

\bibitem{Hoeffding65} \textsc{Hoeffding, W.} (1965).
Asymptotically optimal tests for multinomial distributions.
\textit{Ann. Math. Statist.} \textbf{36} (2) 369--401.

\bibitem{Holevo73} \textsc{Holevo, A. S.} (1973).
Statistical decision theory for quantum systems.
\textit{J. Multivariate Anal.} \textbf{3} 337--394.

\bibitem{Holevo78} \textsc{Holevo, A. S.} (1978).
On asymptotically optimal hypothesis testing in quantum statistics.
\textit{Theor. Prob. Appl.} \textbf{23} 411--415.

\bibitem{KRS09} \textsc{K\"onig, R., Renner, R.} and \textsc{Schaffner, C.} (2009).
The operational meaning of min- and max-entropy.
\textit{IEEE Trans. Inf. Theory} \textbf{55} (9) 4337--4347.

\bibitem{LeangJohnson97} \textsc{Leang, C. C.} and \textsc{Johnson, D. H.} (1997).
On the asymptotics of M-hypothesis Bayesian detection.
\textit{IEEE Trans. Inf. Theory} \textbf{43} (1) 280--282.

\bibitem{Li12} \textsc{Li, K.} (2014).
Second-order asymptotics for quantum hypothesis testing.
\textit{Ann. Statist.} \textbf{42} (1) 171--189.

\bibitem{Mosonyi09} \textsc{Mosonyi, M.} (2009).
Hypothesis testing for Gaussian states on bosonic lattices.
\textit{J. Math. Phys.} \textbf{50} 032105.

\bibitem{MHOF08} \textsc{Mosonyi, M., Hiai, F., Ogawa, T.} and \textsc{Fannes, M.} (2008).
Asymptotic distinguishability measures for shift-invariant quasi-free states of fermionic lattice systems.
\textit{J. Math. Phys.} \textbf{49} 072104.

\bibitem{MosonyiOgawa13} \textsc{Mosonyi, M.} and \textsc{Ogawa, T.} (2015).
Quantum hypothesis testing and the operational interpretation of the quantum Renyi relative entropies.
\textit{Comm. Math. Phys.} \textbf{334} (3) 1617--1648.

\bibitem {Nussbaum13} \textsc{Nussbaum, M.} (2013).
Attainment of the multiple quantum Chernoff bound for certain ensembles of mixed states.
In \textit{Proceedings of the First International Workshop on Entangled Coherent States and Its Application to Quantum Information Science, Tamagawa University, Tokyo, Japan} 77--81.

\bibitem {NussbaumSzkola06} \textsc{Nussbaum, M.} and \textsc{Szko{\l}a, A.} (2009).
The Chernoff lower bound for symmetric quantum hypothesis testing.
\textit{Ann. Statist.} \textbf{37} (2) 1040--1057. arXiv:quant-ph/0607216.

\bibitem {NussbaumSzkola10} \textsc{Nussbaum, M.} and \textsc{Szko{\l}a, A.} (2010).
Exponential error rates in multiple state discrimination on a quantum spin chain.
\textit{J. Math. Phys.} \textbf{51} 072203.

\bibitem {NussbaumSzkola11a} \textsc{Nussbaum, M.} and \textsc{Szko{\l}a, A.} (2011).
Asymptotically optimal discrimination between multiple pure quantum states.
In \textit{Theory of Quantum Computation, Communication and Cryptography. 5th Conference, TQC 2010, Leeds, UK.}
\textit{Lecture Notes in Computer Science} \textbf{6519} 1--8.
Springer, Berlin.

\bibitem {NussbaumSzkola11} \textsc{Nussbaum, M.} and \textsc{Szko{\l}a, A.} (2011).
An asymptotic error bound for testing multiple quantum hypotheses.
\textit{Ann. Statist.} \textbf{39} (6) 3211--3233.

\bibitem{OgawaNagaoka00} \textsc{Ogawa, T.} and \textsc{Nagaoka, H.} (2000).
Strong converse and Stein's lemma in the quantum hypothesis testing.
\textit{IEEE Trans. Inf. Theory} \textbf{46} (7) 2428--2433.

\bibitem{Parth01} \textsc{Parthasarathy, K. R.} (2001).
On consistency of the maximum likelihood method in testing multiple quantum hypotheses.
In \textit {Stochastics in Finite and Infinite Dimensions} 361--377.
Birkh\"auser, Boston.

\bibitem{Qiu08} \textsc{Qiu, D. W.} (2008).
Minimum-error discrimination between mixed quantum states.
\textit{Phys. Rev. A} \textbf{77} 012328.

\bibitem{Salikh73} \textsc{Salihov, N. P.} (1973).
Asymptotic properties of error probabilities of tests for distinguishing between several multinomial testing schemes.
\textit{Dokl. Akad. Nauk SSSR} \textbf{209} 54--57.

\bibitem{Salikh98} \textsc{Salihov, N. P.} (1998).
On a generalization of {C}hernoff distance.
\textit{Teor. Veroyatn. Primen.} \textbf{43} 294--314.
Translation in \textit{Theory Probab. Appl.} \textbf{43} (1999) 239--255.

\bibitem{TomamichelHayashi12} \textsc{Tomamichel, M.} and \textsc{Hayashi, M.} (2013).
A hierarchy of information quantities for finite block length analysis of quantum tasks.
\textit{IEEE Trans. Inf. Theory} \textbf{59} (11) 7693--7710.

\bibitem{Torgersen81} \textsc{Torgersen, E. N.} (1981).
Measures of information based on comparison with total information and with total ignorance.
\textit{Ann. Statist.} \textbf{9} 638--657.

\bibitem{Tyson09} \textsc{Tyson, J.} (2009).
Two-sided estimates of minimum-error distinguishability of mixed quantum states via generalized Holevo-Curlander bounds.
\textit{J. Math. Phys.} \textbf{50} 032106.

\bibitem{YKL75} \textsc{Yuen, H. P., Kennedy, R. S.} and \textsc{Lax, M.} (1975).
Optimum testing of multiple hypotheses in quantum detection theory.
\textit{IEEE Trans. Inf. Theory} \textbf{21} (2) 125--134.

\end{thebibliography}
\end{document}